\documentclass{article}
\usepackage[utf8]{inputenc}
\usepackage{color} 

\title{Compressing continuous variable quantum measurements}
\usepackage{amsmath,amssymb,amsthm}
\usepackage{latexsym}
\usepackage{amsfonts}
\usepackage{bbm,bm,dsfont} % for bold greek symbols, write \bm{\sigma} for example
\usepackage{graphicx}
\usepackage{authblk}

\usepackage{soul}

\numberwithin{equation}{section} %mit� liene t�m� tarkoittaa?

\theoremstyle{definition}
\newtheorem{proposition}{Proposition}
\newtheorem{definition}{Definition}
\newtheorem{example}{Example}

\newtheorem{theorem}{Theorem}
\newtheorem{lemma}{Lemma}

\newcommand{\hA}{\mathcal{A}}
\newcommand{\hB}{\mathcal{B}}

\newcommand{\hD}{\mathcal{D}}
\newcommand{\hE}{\mathcal{E}}

\newcommand{\hL}{\mathcal{L}}
\newcommand{\hM}{\mathcal{M}}

\newcommand{\hP}{\mathcal{P}} %projections

\newcommand{\hS}{\mathcal{S}}
\newcommand{\hT}{\mathcal{T}}

\newcommand{\N}{\mathbb N} %natural
\newcommand{\R}{\mathbb R} %real
\newcommand{\C}{\mathbb C} %complex
\newcommand{\fii}{\varphi} % hassu fii
\newcommand{\hi}{\mathcal{H}} %Hilbert space
\newcommand{\ki}{\mathcal{K}} %other Hilbert space
\newcommand{\li}{\mathcal{L}} %3rd Hilbert space
\newcommand{\id}{\mathds1} %identity operator, MIK� OIS HYV�? PAKSU YKK�NEN? {\openone} ei toimi nyt!
\newcommand{\lh}{\mathcal{L(H)}} %bounded linear operators on H
\renewcommand{\th}{\mathcal{T(H)}} %trace class operators on H
\newcommand{\tk}{\mathcal{T(K)}} %trace class operators on K
\newcommand{\sh}{\mathcal{S(H)}} %states on H
\newcommand{\tr}[1]{\mathrm{tr}\left[#1\right]} %trace
\def\<{\langle} %mainiot sulut! <
\def\>{\rangle} %mainiot sulut! >
\newcommand{\ket}[1]{|#1\rangle} %ket
\newcommand{\bra}[1]{\langle #1|} %bra
\newcommand{\kb}[2]{|#1 \rangle\langle #2|} %ketbra
\newcommand{\ip}[2]{\left\langle #1 | #2 \right\rangle} % sis�tulo (huom! V�lipalkki | j�� pieneksi)
\newcommand{\bo}[1]{\mathcal{B}\left(#1\right)} %yleinen Borel sigma-algebra
\newcommand{\lin}{\mathrm{lin}\,} %linear combinations
\newcommand{\CHI}[1]{\ensuremath{ \chi\raisebox{-1ex}{$\scriptstyle #1$} }} % karakteristinen funktio, esim. \CHI X
\begin{document}

\author[1]{Pauli Jokinen}
\author[1]{Sophie Egelhaaf}
\author[2]{Juha-Pekka Pellonpää}
\author[1]{Roope Uola}
\affil[1]{Department of Applied Physics, University of Geneva, 1211 Geneva, Switzerland}
\affil[2]{Department of Physics and Astronomy, University of Turku, FI-20014 Turun yliopisto, Finland}

\maketitle

\begin{abstract}
    We generalize the notion of joint measurability to continuous variable systems by extending a recently introduced compression algorithm of quantum measurements to this realm. The extension results in a property that asks for the minimal dimensional quantum system required for representing a given set of quantum measurements. To illustrate the concept, we show that the canonical pair of position and momentum is completely incompressible. We translate the concept of measurement compression to the realm of quantum correlations, where it results in a generalisation of continuous variable quantum steering. In contrast to the steering scenario, which detects entanglement, the generalisation detects the dimensionality of entanglement. We illustrate the bridge between the concepts by showing that an analogue of the original EPR argument is genuinely infinite-dimensional with respect to our figure of merit, and that a fundamental discrete variable result on preparability of unsteerable state assemblages with separable states does not directly carry over to the continuous variable setting. We further prove a representation result for partially entanglement breaking channels that can be of independent interest.
\end{abstract}

\section{Introduction}

A key concept in quantum measurement theory is that of joint measurability. This generalization of commutativity is a central tool in the foundations of quantum theory, as shown by, e.g., the resolution of the Heisenberg microscope \cite{buschrmp2014}, Fine’s theorem \cite{fine82}, and connections to fundamental bounds in interferometry \cite{Kiukas22coherence}. On the applied side, the research on the topic has been recently fueled by bridges between incompatibility of measurements and the advantage provided in various quantum correlation tasks, such as Bell non-locality \cite{wolf09,brunner14}, EPR steering \cite{wiseman07,cavalcanti17,uola20review,quintino14,uola14,uola15}, contextuality \cite{tavakoli19,xu19}, and temporal correlations \cite{uola19a,uola2022retrievability}, as well as in operational tasks such as quantum state discrimination \cite{carmeli19a,skrzypczyk19,oszmaniec19,uola19b,uola19c}, and programmability \cite{buscemi20}, see \cite{Guhne2023JMreview} for a review.

Remarkably, many of the recent results and directions have concentrated on the finite-dimensional and especially on the discrete variable regime. In this manuscript, we take one of such operational directions, that of compressibility, or $n$-simulability, and present a proper generalization of it to the continuous variable realm. This notion asks whether one could simulate a given set of measurements by first distributing a quantum state of a $d$-dimensional system into smaller $n$-dimensional systems and then performing measurements on these smaller systems, cf.\  Fig.~\ref{fig:CompressionDiscreteFinite}. In the case $n=1$, one gets the notion of joint measurability. This is due to the fact, that the distribution of the state becomes a quantum-to-classical map, i.e.\ a measurement, and the measurements on the smaller one-dimensional systems are simply classical data processings.

\begin{figure}
    \centering
    \includegraphics[width= 0.9\linewidth, clip]{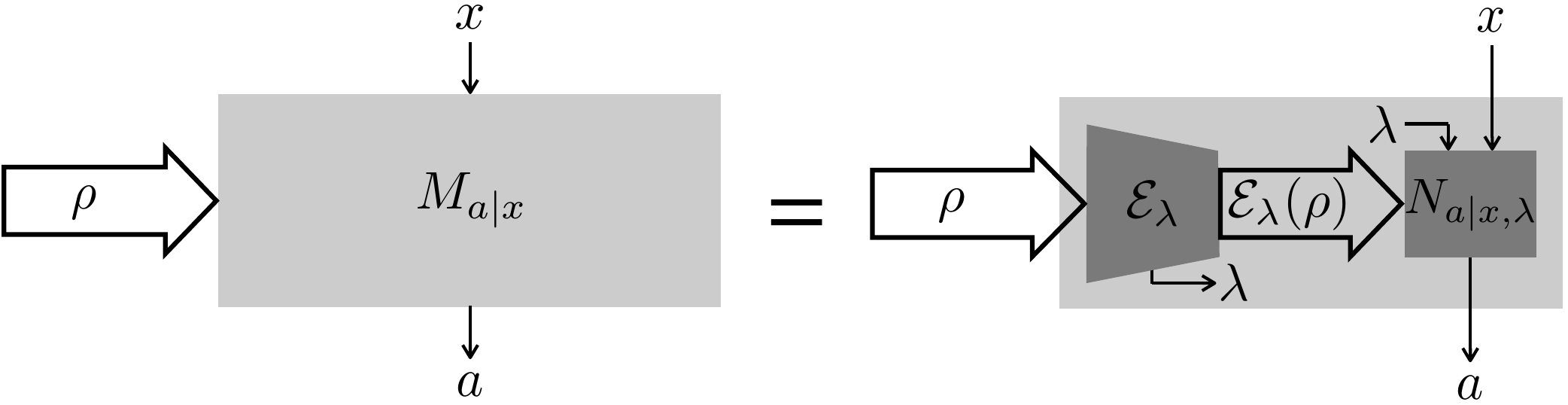}
    \caption{If a set of measurements $\{M_{a|x}\}_{a,x}$ is $n$-compressible, its statistics $p(a|x)$ can be simulated for any state $\rho \in \hS\left(\C^d\right)$ by first applying an instrument $\hE_\lambda$ which returns a set of sub-normalised states $\hE_\lambda(\rho) \in \hT\left(\C^n\right)$ labeled by $\lambda$ and then acting on these states with a set of measurements $\{N_{a|x,\lambda}\}_{a,x,\lambda}$.}
    \label{fig:CompressionDiscreteFinite}
\end{figure}

We show that in order to do the generalisation in line with the finite-dimensional compression scheme, one needs to take the non-trivial step of changing from discrete Kraus decompositions of quantum channels into continuous or point-wise decompositions. This further requires one to change from traditional bounded Kraus operators to unbounded ones. Notably, this provides an operational application of Heisenberg instruments whose output space is not the full quantum algebra. The process also involves proving a structural result on partially entanglement breaking channels, that is characterizing such maps with their pointwise Kraus decompositions. After arriving at a definition that is consistent with joint measurability in the limiting $n=1$ case, we show through an example that position and momentum are not finitely compressible.

We further generalize a recently found connection between $n$-simulability and high-dimensional quantum steering \cite{jones2022equivalence}. In the finite-dimensional case, such connection has been fruitful, in that it has been used to introduce entanglement-assisted local models \cite{jones2022equivalence} and it has shown a deep connection between $n$-simulability and recent demonstrations of high-dimensional entanglement in the semi-device independent regime \cite{designolle21}.  We present a general definition of high-dimensional steering, or $n$-preparability, which is consistent with finite dimensional case. Furthermore, we show that the fundamental result stating that any unsteerable state assemblage is preparable with some separable state \cite{kogias15b,moroder16} has to be treated carefully when applied to infinite dimensions.
Moreover, we use the result on non-compressibility of position and momentum to show that an analogue of the original EPR argument is not $n$-preparable for any finite $n$.

The paper is structured as follows. In the second section, we present the problem of measurement compression in the finite-dimensional case using Kraus decompositions of quantum channels. In the third section, we propose a definition for compressibility in the continuous variable setting. We show that in order to preserve the finite-dimensional operational structure, one needs to consider a more general formulation consisting of pointwise (unbounded) Kraus representations or, equivalently, direct integral Hilbert spaces. This section also includes some structural results on $n$-PEB channels. Here, we further show that position and momentum are non-finitely compressible. In section four, we build a bridge between $n$-simulability and high-dimensional steering for the general case and show how our results on measurement compression can be translated to the shared scenario.

\section{Simulating finite-dimensional quantum measurements}

We begin by describing the measurement compression scenario in the finite dimensional framework originally introduced in \cite{ioannou2022simulability}. Suppose we have a set of quantum measurements $\hM=\{M_{a|x}\}_{a,x}$ in a $d$-dimensional Hilbert space, where each $\{M_{a|x}\}_a$ is a POVM (normalized positive operator-valued measure) i.e.\ a set of positive semi-definite linear operators with the normalization condition $\sum_{a} M_{a|x}=\id$. Throughout this article we let $\id$ denote the identity operator of any linear space. This set of measurements produces probabilities on an arbitrary quantum state $\rho$, i.e a linear operator for which $\rho\geq0$  and $\tr{\rho}=1$, via the Born rule $p_{\rho}(a|x)=\tr{M_{a|x}\rho}$. We denote the state space (resp.\ trace-class) of a Hilbert space $\hi$ by $\sh$ (resp.\ $\th$). 
In addition, $\li(\hi)$ denotes the bounded linear operators on $\hi$.
If $\hi$ is finite dimensional, we usually identify $\hi$ with $\C^d$ where $d=\dim\hi$. In this case, $\mathcal T(\C^d)=\mathcal L(\C^d)$.

In the measurement compression scenario we look to recover the exact statistics of $\hM$ with POVMs operating in $n$-dimensional (or smaller) Hilbert space with $n<d$. Formally the concept of compressibility is defined by utilizing a quantum instrument $\{\hE_{\lambda}\}_{\lambda}$, i.e.\ a set of completely positive trace non-increasing maps $\hE_{\lambda}:\mathcal T(\C^d) \to \mathcal T(\C^n)$ such that $\sum_{\lambda} \hE_{\lambda}$ is trace-preserving.  Specifically, a set of measurements $\hM=\{M_{a|x}\}_{a,x} \subseteq \li(\C^d)$ is \emph{$n$-simulable} if and only if there exist an instrument $\{\hE_{\lambda}\}_{\lambda}$ and a set of POVMs $\{N_{a|x,\lambda}\}_{a,x,\lambda} \subseteq \li(\C^n)$ such that the statistics of $\hM$ in an arbitrary state $\rho\in\mathcal S(\C^d)$ can be recovered from the $n$-dimensional POVMs mediated by the instrument:
\begin{align}\label{Eq:simfinite}
    \tr{\rho M_{a|x}}=\sum_{\lambda} \tr{\hE_{\lambda}(\rho) N_{a|x,\lambda}}.
\end{align}
This equation implies that $\hM$ is 1-simulable if and only if it is jointly measurable, i.e.\ $M_{a|x}=\sum_{\lambda} p(a|x,\lambda) G_{\lambda}$ for all $a,\,x$, where $\{G_{\lambda}\}_{\lambda}$ is some ``parent'' POVM. This is one of the equivalent definitions of joint measurability for discrete-outcome measurements. The equivalence with joint measurability is also intuitively plausible, since in the case $n=1$ the full quantum information is compressed to purely classical. Therefore the concept of $n$-simulability is a relaxation of the concept of joint measurability.

The definition of $n$-simulability is known \cite[Claim 5]{ioannou2022simulability} to be related to \emph{n-partially entanglement breaking} channels ($n$-PEB), i.e.\ completely positive trace preserving (CPTP) maps $\Lambda$ such that, for a Hilbert space $\hi$ and for an arbitrary $\rho \in \hS(\C^d \otimes \hi)$, we have  $\mathrm{SN}[\Lambda \otimes \id_{\hi}(\rho)]\leq n$. Here $\mathrm{SN}(\rho)=\inf_{\sum_i p_i \kb{\psi_i}{\psi_i}=\rho} \sup_i \mathrm{SR}(\kb{\psi_i}{\psi_i})$ is the Schmidt number, where $\mathrm{SR}$ is the Schmidt rank of a pure state, i.e.\ the dimension of the support of either subsystem's state. Specifically the $n$-simulable sets of POVMs are exactly those for which there exist a $n$-PEB channel $\Lambda: \li(\C^d) \to \li(\hi)$ and a set of POVMs $\{N_{a|x}\}_{a,x} \subseteq \li(\hi)$ such that $M_{a|x}=\Lambda^*(N_{a|x})$ for all $a$ and $x$. Here $\Lambda^*$ denotes the Heisenberg picture of the Schrödinger channel $\Lambda$. In the finite dimensional case, a $n$-PEB channel $\Lambda$ is known to possess a rank-restricted Kraus decomposition, i.e.\ $\Lambda(\rho)=\sum_{\lambda} K_{\lambda} \rho K_{\lambda}^*$ with $\mathrm{rank}(K_{\lambda})\leq n$ for all $\lambda$ \cite{Chruscinski2006}. Therefore the finite-dimensional quantum framework allows us to define the concept of $n$-simulability essentially via the Kraus decomposition and the ranks of the Kraus operators of the related channel.

\section{Continuous variable compression}

In this section we aim to generalize the concept of simulability to continuous-variable measurements operating in infinite-dimensional \emph{separable} Hilbert spaces by using Kraus-type representations. A separable Hilbert space is a Hilbert space containing a countable dense set (or a countable basis), and going forward we will assume that all Hilbert spaces are of this type unless otherwise stated. As it will turn out, the generalization of simulability into the continuous variable case requires richer structure than the finite case to be consistent. \par 

As mentioned, for a consistent generalization, 1-simulable sets of measurements should be equivalent to jointly measurable ones. A quantum measurement $M$ in the continuous variable setting is also described as a POVM. Specifically, letting $\Omega$ be the space of outcomes of the measurement and $\hA$ be the set of possible outcome subsets of $\Omega$, i.e.\ a suitable $\sigma$-algebra of $\Omega$, a POVM is a map $M:\hA \to \li(\hi)$ that is $\sigma$-additive in the weak operator topology, $M(X)\geq 0$ for all $X \in \hA$ and $M(\Omega)=\id$. For a discrete $\sigma$-algebra, this reduces to the usual definition of a discrete POVM. For the rest of the paper we assume the measurable spaces to be \emph{standard Borel} i.e.\ Polish spaces equipped with their Borel $\sigma$-algebras (this can be slightly generalized, see \cite{Pello7}). Joint measurability is then defined as a straightforward generalization of the definition in the discrete setting: A set of measurements $\{M_x\}_x$ is jointly measurable if and only if there exists a ``parent'' POVM $G$ defined in the measurable space $(\Omega,\hA)$ and a post-processings, i.e.\ \emph{Markov kernels} $p(\cdot|x,\cdot):\hA_x \times \Omega \to [0,1]$ such that, for all $x$, 
\begin{align}
    M_x(X_x)=\int_{\Omega} p(X_x|x,\lambda) \, dG(\lambda) \qquad\text{(weakly)}
\end{align}
for all $X_x \in \hA_x$.\footnote{Clearly, the maps $\lambda\mapsto p(X_x|x,\lambda)$ must be $G$-measurable.} 
Here the $(\Omega_x,\hA_x)$ describe the outcomes for the setting $x$. Equivalently joint measurability can be defined as the existence of a marginal form joint observable \cite[Theorem 11.1]{busch16}.  For an in-depth look about continuous-variable POVMs and operator integrals, see \cite{busch16}. \par 
We now ask whether the definitions for simulability of the discrete measurements acting on finite dimensional Hilbert spaces can be generalized while being consistent with joint measurability. We begin by tackling the channel dependent definition. Let $\hM=\{M_x\}_x$ be a set of measurements. Does it hold that $\hM$ is jointly measurable if and only if there is an entanglement breaking (1-PEB) channel  $\Lambda: \th \to \tk$ and a set of POVMs $\{N_x\}_x$ acting on $\ki$, where $\ki$ is another Hilbert space, such that $\Lambda^*(N_x(X_x))=M_x(X_x)$ for all $X_x \in \hA_x$ and $x$? \par 

To do so, we first need to consider $n$-PEB channels in infinite dimensions a bit more carefully. One needs to use the generalization of the finite-dimensional Schmidt number introduced in \cite{shirokov2013schmidt}. This is defined as follows: Let $\omega \in \hS(\hi \otimes \ki)$ and let $\hP_{\omega}(\mathrm{extr}(\hS(\hi \otimes \ki))$ be the set of probability measures defined in the Borel $\sigma$-algebra of the extremals (i.e.\ pure states) of $\hS(\hi \otimes \ki)$ with the topology induced by the trace norm such that for all $\mu \in \hP_{\omega}(\mathrm{extr}(\hS(\hi \otimes \ki))$ we have $\omega=\int_{\mathrm{extr}(\hS(\hi \otimes \ki))} \rho \, d\mu(\rho)$. In this case $\omega$ is called a \emph{barycenter} of $\mu$. Then the Schmidt number is defined as  $\mathrm{SN}(\omega):=\inf_{\mu \in \hP_{\omega}(\mathrm{extr}(\hS(\hi \otimes \ki))} \mu\text{-ess} \sup \{\mathrm{SR}(\rho) \,|\,\omega=\int\rho \, d\mu(\rho)\}$.  The need for this generalization draws its roots from the existence of non-countably decomposable separable states \cite{holevo2005separability}. The $n$-PEB property of a channel $\Lambda$ is then defined as the requirement that $\mathrm{SN}(\Lambda \otimes \id (\rho)) \leq n$ for an arbitrary state $\rho$. \par 

The ``if''-direction, i.e.\ EBC-preparable $\Rightarrow$ jointly measurable, follows by a direct calculation from Theorem 2 of \cite{holevo2005separability}. However, the ``only if''-direction is more complicated. It turns out that a simple example of position measurement with itself, seen as a jointly measurable pair, already demonstrates that a richer structure than EBC-preparability is needed. This can be seen as a consequence of the following result.

\begin{theorem} \label{thmseparability}
    Let $M$ be a continuous extremal POVM defined in the standard Borel measurable space ($\Omega, \hB(\Omega)$) and operating in $\lh$. If there exists a POVM $N$ operating in $ \li(\ki)$ and an EB-channel $\Lambda:\th \to \tk$ such that $M(X)=\Lambda^*(N(X))$ for all  $X \in \hA$, then the Hilbert space $\ki$ cannot be separable.
\end{theorem}
\begin{proof}
    
Suppose now that there is a  POVM $N: \hB(\Omega) \to \mathcal{L}(\ki)$ and an EB-channel $\Lambda:\th \to \mathcal{T}(\ki)$ such that $\Lambda^*(N(X))=M(X)$ for all $X \in \hB(\Omega)$. Using Theorem 2 of \cite{holevo2005separability}, this implies that, for all $X$, we have  $M(X)=\int_{\Omega'} \tr{\sigma_{\lambda}N(X)} dG(\lambda)$ for some POVM $G$ in the measurable space $(\Omega',\hB(\Omega'))$ and measurable family of states $\sigma_{\lambda}$. 

 As the POVM $M$ is extremal, we have by \cite[Theorem 3.3]{Jencova2007} that each $(X,\lambda) \mapsto \tr{\sigma_{\lambda}N(X)}$ must be $\{0,1\}$-valued almost everywhere. Since we are dealing with standard Borel spaces,  we have a measurable function $f:\Omega' \to \Omega$ such that $\tr{\sigma_{\lambda}N(X)}=\delta_{f(\lambda)}(X)$ almost everywhere \cite[Theorem 3.3]{Jencova2007}. We denote some countable generating ring of the $\sigma$-algebra $ \hB(\Omega)$ by $\{X_k\}_{k \in \N}$. Then we have $\tr{\sigma_{\lambda}N(X_k)}=\delta_{f(\lambda)}(X_k)$ for $\lambda \notin C_k$, where $G(C_k)=0$. Letting then $C:=\bigcup_k C_k$, we have $G(C)=0$ and for all $n$ the equality $\tr{\sigma_{\lambda}N(X_n)}=\delta_{f(\lambda)}(X_n)$  holds for $\lambda \notin C$. Then by Caratheodory's extension theorem we have  $\tr{\sigma_{\lambda}N(X)}=\delta_{f(\lambda)}(X)$ for all $X$ and $\lambda \notin C$, since the extension is unique for finite measures. We can then define a new measurable space as the standard Borel space ($\Omega'\setminus C,\hB(\Omega'\setminus C))$. Consequently we assume that $\tr{\sigma_{\lambda}N(\cdot)}$ is $\delta_{f(\lambda)}$ everywhere.

Now using this we find that 
\begin{align}
N(X)=\delta_{f(\lambda)}(X)P_{\lambda}+ P_{\lambda}^{\perp}N(X)P_{\lambda}^{\perp}, \label{NX}
\end{align} 
where $P_{\lambda}$ is the support projection of $\sigma_{\lambda}$. This is seen as follows. Let $0\leq E\leq \id$ be such that $\tr{\sigma_{\lambda}E}=0$. Then obviously $P_{\lambda}E P_{\lambda}=0$ and thus  $P^{\perp}_{\lambda}E P_{\lambda}=P_{\lambda}E P^{\perp}_{\lambda}=0$ yielding  $E=P^{\perp}_{\lambda}EP^{\perp}_{\lambda}$. Now since $\tr{\sigma_{\lambda}N(X)}$ is $\{0,1\}$-valued, for a fixed $X$ we only have the two cases $\tr{\sigma_{\lambda}N(X)}=0$ or $\tr{\sigma_{\lambda}(\id-N(X))}=0$. Thus on one hand if $f(\lambda) \notin X$ then $N(X)=P^{\perp}_{\lambda}N(X)P^{\perp}_{\lambda}$. On the other hand if $f(\lambda) \in X$, we have $\id-N(X)=P^{\perp}_{\lambda}(\id-N(X))P^{\perp}_{\lambda}$, which implies $N(X)=P_{\lambda}+P^{\perp}_{\lambda}N(X)P^{\perp}_{\lambda}$. In both cases equation (\ref{NX}) holds. \par 

Let us then choose $\lambda'\neq \lambda$ such that $f(\lambda)\neq f(\lambda')$. Then if $X=\Omega\setminus\{f(\lambda')\}$ (which is measurable), we have 
\begin{align}
0=\tr{\sigma_{\lambda'}N(X)}=\tr{\sigma_{\lambda'}P_{\lambda}}+\tr{P_{\lambda}^{\perp}\sigma_{\lambda'}P_{\lambda}^{\perp}N(X)}. \label{unc}
\end{align}
Since both of the terms on the RHS of (\ref{unc}) are positive, we have especially $\tr{\sigma_{\lambda'}P_{\lambda}}=0$. This implies that $P_{\lambda}P_{\lambda'}=0$. Therefore, if the image of $f$ is uncountable, there is an uncountable amount of orthogonal projections, which would mean that the space $\ki$ is not separable. \par We thus show that the image of $f$ must be uncountable. Aiming for a contradiction, suppose that the image is countable: $f(\Omega')=\{x_1,x_2,x_3,\dots\}$. Now by Theorem 3.3 of \cite{Jencova2007} we have that $M(X)=G(f^{-1}(X))$. Therefore 
\begin{align*}
    M(X)=G(f^{-1}(X)) =G(\Omega' \cap f^{-1}(X)) \leq G(f^{-1}(f(\Omega') \cap X))=\sum_i G(f^{-1}(\{x_i\} \cap X))
\end{align*}
Thus $M$ is a discrete POVM, which is a contradiction, and the proof is finished.
\end{proof}
Let us now show how this relates to joint measurability via a concrete example. For this, let $L^2(\R)$ be the Lebesgue space of square integrable functions and $Q: \hB(\R) \to \hL(L^2(\R))$ the spectral measure of the canonical position operator defined by $(Q(X)\fii)(\lambda)=\CHI{X}(\lambda)\fii(\lambda)$ for almost all $\lambda \in \R$. Here $\CHI{X}$ denotes the characteristic function of the Borel set $X \in \hB(\R)$. Define $Q_1=Q_2=Q$ so that $(Q_1,Q_2)$ is a jointly measurable pair, a possible joint POVM being $G(X_1\times X_2)=Q(X_1\cap X_2)$, $X_1,\,X_2\in\hB(\R)$. As $Q$ is projection valued it is especially an extremal POVM \cite{Pello_extreme} and furthermore it has an empty discrete spectrum. Thus by Theorem \ref{thmseparability} the EBC-preparability condition, $Q_i(X)=\Lambda^*(N_i(X))$ for some POVM pair $(N_1,N_2)$, would imply that the POVMs $N_i$ operate in a non-separable Hilbert space.

Our aim is to define a notion of $n$-simulability that would reduce to joint measurability in the case $n=1$. Based on the above result, we are forced to go beyond the $n$-PEB definition. It turns out that the following definition, which is a generalization of the instrument-based definition in Eq.~(\ref{Eq:simfinite}), has the desired property.

\begin{definition}\label{defsimulability}
Let $\hM:=\{M_x\}_x$ be a set of POVMs with $M_x$ defined in the measurable space $(\Omega_x,\hA_x)$ operating in a separable Hilbert space. Then $\hM$ is \emph{$n$-simulable} if and only if there exists a probability space $(\Omega,\hA,\mu)$, a weakly measurable family $\{K_{\lambda}:\hD \to \ki\}_{\lambda \in \Omega}$ of possibly unbounded operators with a common dense domain $\hD\subseteq\hi$, and the rank restriction $\mathrm{rank}(K_{\lambda})\leq n$ for $\mu$-almost all $\lambda$, and a weakly measurable family of POVMs $\{N_{x,\lambda}:\hA_x\to\li(\ki)\}_{x,\lambda}$ such that for all $\fii,\psi \in \hD$,  $x$ and $X_x\in \hA_x$ we have 
\begin{align}
    \ip{\fii}{M_x(X_x)|\psi}=\int_{\Omega} \ip{\fii}{K_{\lambda}^* N_{x,\lambda}(X_x) K_{\lambda}|\psi} \, d\mu(\lambda)
\end{align}
\end{definition}

The following two results show that this definition is natural and consistent with the requirements. The first one also serves as a useful characterization of $n$-PEB channels in the infinite dimensional continuous variable setting. \par 
It was shown \cite[Proposition 9]{shirokov2013schmidt} that there exist $n$-PEB channels that can only have infinite-rank Kraus decompositions. However, it turns out to be useful for the purpose of simulability of measurements to prove the following result that characterizes $n$-PEB channels using a \emph{pointwise Kraus decomposition}, similar to Definition \ref{defsimulability}. That is, the following in conjuction with Proposition \ref{thmjointmeas} shows that the functions $\lambda \mapsto N_{x,\lambda}$ generally need not be constant.
\begin{proposition}\label{thmnpeb}
A quantum channel $\Lambda:\th \to \tk$ is $n$-partially entanglement breaking if and only if there exists a probability space $(\Omega,\hA,\mu)$, a weakly measurable family $\{K_{\lambda}\}_{\lambda \in \Omega}$ of possibly unbounded operators with a common dense domain $\hD\subseteq\hi$, and the rank restriction $\mathrm{rank}(K_{\lambda})\leq n$ for $\mu$-almost all $\lambda$, such that for all $\fii,\psi \in \hD$ and $A \in \li(\ki)$ we have 
\begin{align}
    \ip{\fii}{\Lambda^*(A)|\psi}=\int_{\Omega} \ip{\fii}{K_{\lambda}^* A K_{\lambda}|\psi} \, d\mu(\lambda)
\end{align}
\end{proposition}

\begin{proof}
    ($\Rightarrow$) Let $\Lambda$ be $n$-PEB and fix an orthonormal basis $\{\ket{k}\}_{k=1}^{\dim\hi}$. We define a faithful full rank state $\sigma:= \sum_{k=1}^{\infty} a_k \kb kk$, where $a_k>0$ for all $k$ and $\sum_k a_k=1$. Let then $\ket{\psi_{\sigma}}=\sum_{k=1}^{\dim\hi} \sqrt{a_k}\ket{k}\otimes\ket{k}$ be a purification of $\sigma$ in $\hi \otimes \hi$. As $\Lambda$ is $n$-PEB, we have $\Lambda \otimes \id (\kb{\psi_{\sigma}}{\psi_{\sigma}}) \in \hS_n(\ki \otimes \hi)$, where $\hS_n(\ki \otimes \hi):=\{ \rho \in \hS(\ki \otimes \hi) \, \, | \, \, \mathrm{SN}(\rho)\leq n\}$. We denote the pure states in this set by $\hS_n^1(\ki \otimes \hi)$. By Proposition 1B in \cite{shirokov2013schmidt} and Lemma 1 in \cite{holevo2005separability} we have the following with some abuse of notation:
    \begin{align}
        \Lambda \otimes \id (\kb{\psi_{\sigma}}{\psi_{\sigma}})=\int_{\hS_n^1(\ki \otimes \hi)} \kb{\fii}{\fii} \, d\mu(\fii) \label{eq1}
    \end{align}
    Here $\mu$ is a Borel probability measure. \par 
    Let then $T \in \th$ be arbitrary. With a direct calculation we find that 
    \begin{align}
        \Lambda(\sigma^{1/2}T\sigma^{1/2})=\mathrm{tr}_{\hi}[(\id \otimes T^{\rm T})\Lambda \otimes \id (\kb{\psi_{\sigma}}{\psi_{\sigma}})] \label{eq2}
    \end{align}
    Here the transpose $T^{\rm T}$ is in the fixed basis $\{\ket{k}\}_{k=1}^{\dim\hi}$. By combining equation (\ref{eq1}) and (\ref{eq2}) we aim to find the desired integral form. For this fix an orthonormal basis $\{\ket{\eta_m}\}_{m=1}^{\dim\ki}$ for $\ki$ and, for each $\fii\in\ki\otimes\hi$, define the Hilbert-Schmidt operators $F_{\fii}:\hi \to \ki$ by the relations $\ip{\eta_m}{F_{\fii}|k}=\ip{\eta_m\otimes k}{\fii}$. 
    If $\kb\fii\fii \in \hS_n^1(\ki \otimes \hi)$ then
     $\mathrm{SR}(\kb\fii\fii)\leq n$, and we have that $\mathrm{rank}(F_{\fii})\leq n$. The mapping $\fii\mapsto F_{\fii}$ is also clearly (weakly) continuous. Now the following holds weakly: 
    \begin{align}
        \kb{\fii}{\fii}=\sum_{i,j,k,l=1}^{\infty} \ip{\eta_i}{F_{\fii}|j}\ip{k}{F_{\fii}^*|\eta_l}\kb{\eta_i\otimes j}{\eta_l\otimes k}
    \end{align}
    This implies that
    \begin{align}
        \mathrm{tr}_{\hi}[(\id \otimes T^{\rm T})\kb{\fii}{\fii}]=\sum_{i,j,k,l=1}^{\infty} \ip{\eta_i}{F_{\fii}|j}\ip{k}{F_{\fii}^*|\eta_l}\ip{j}{T|k}\kb{\eta_i }{\eta_l}=F_{\fii}T F_{\fii}^*
    \end{align}
    Combining this with equations (\ref{eq1}) and (\ref{eq2}) we get\footnote{Note that, if $\kb\fii\fii=\kb\psi\psi$, then $\psi=c\fii$, $c\in\C$, $|c|=1$, and $F_\psi=c F_\fii$ but $F_{\fii}T F_{\fii}^*=F_{\psi}T F_{\psi}^*$.} 
    \begin{align}
        \Lambda(\sigma^{1/2}T\sigma^{1/2})=\int_{\hS_n^1(\ki \otimes \hi)} F_{\fii}T F_{\fii}^* \, d\mu(\fii)
    \end{align}
    Let us then define the unbounded operators $K_{\fii}:=F_{\fii}\sigma^{-1/2}$. These are densely defined in a common dense domain ${\rm ran}\,\sigma^{1/2}$ due to $\sigma$ being a faithful full rank state. Then in the Heisenberg picture we have densely $\Lambda^*(A)=\int_{\hS_n^1(\ki \otimes \hi)} K_{\fii}^*A K_{\fii} \, d\mu(\fii) $. This gives us the first direction.
    \par 
    ($\Leftarrow$) Let us assume the pointwise Kraus representation for the channel $\Lambda$. We aim to use Proposition 8 in \cite{shirokov2013schmidt} and therefore we look to define a suitable state $\ket{\eta}$ with full rank partial states. For this, let $\{\ket{k}\}_{k=1}^{\dim\hi}$ be an orthonormal basis inside the dense domain $\hD$ and $(a_k)$ a sequence of strictly positive numbers with $\sum_k a_k^2=1$. With this we define the following.
    \begin{align}
        \ket{\eta}&:=\sum_{k=1}^{\dim\hi} a_k \ket k\otimes\ket k \\
        \ket{\eta_m}&:=\frac{1}{\sqrt{\sum_{k=1}^m a_k^2}}\sum_{k=1}^m a_k \ket k\otimes\ket k
    \end{align}
    Here $\lim_{m\to \infty} \kb{\eta_m}{\eta_m} =\kb{\eta}{\eta}$ in the trace norm. \\
    We first show that $(\Lambda \otimes \id)(\kb{\eta_m}{\eta_m}) \in \hS_n(\ki \otimes \hi)$ for all $m \in \N$. To this aim, we calculate $\Lambda (\kb{j}{k})$. Let $\fii,\psi \in \ki$ be arbitrary.
    \begin{align}
        \ip{\fii}{\Lambda(\kb{j}{k})|\psi}&=\tr{\kb{j}{k}\Lambda^*(\kb{\psi}{\fii})}=\int_{\Omega} \ip{k}{K_{\lambda}^*|\psi}\ip{\fii}{K_{\lambda}|j} \, d\mu(\lambda)
    \end{align}
    Therefore $\Lambda(\kb{j}{k})=\int_{\Omega} K_{\lambda}\kb{j}{k}K_{\lambda}^* \, d\mu(\lambda)$. Here the $\lambda \mapsto K_{\lambda}\kb{j}{k}K_{\lambda}^*$ is trace-class-Bochner-integrable, since it is Bochner measurable by assumption and also we have 
    \begin{align}
        \int_{\Omega} \Vert K_{\lambda}\kb{j}{k}K_{\lambda}^* \Vert_1 \, d\mu(\lambda) &= \int_{\Omega} \Vert K_{\lambda}\ket{j} \Vert \Vert K_{\lambda}\ket{k} \Vert \, d\mu(\lambda) \\
        &\leq \left(\int_{\Omega} \Vert K_{\lambda}\ket{j} \Vert^2 \, d\mu(\lambda)\right)^{1/2}\left(\int_{\Omega} \Vert K_{\lambda}\ket{k} \Vert^2 \, d\mu(\lambda)\right)^{1/2}=1
    \end{align}
    where we denote the trace norm by $\Vert \cdot \Vert_1$. 
    Here we used the fact that densely $\int_{\Omega} K_{\lambda}^* K_{\lambda} \, d\mu(\lambda)=\id$.
    Let then $C_m:=\frac{1}{\sqrt{\sum_{k=1}^m a_k^2}}$ so that 
    \begin{align}
        (\Lambda \otimes \id)(\kb{\eta_m}{\eta_m})&=C_m^2\sum_{j,k=1}^{m} a_j a_k \Lambda(\kb{j}{k}) \otimes \kb{j}{k} \\
        &=\int_{\Omega} (K_{\lambda} \otimes \id)\left(C_m^2\sum_{j,k=1}^{m} a_j a_k \kb{j}{k} \otimes \kb{j}{k}  \right)(K_{\lambda}^* \otimes \id) \, d\mu(\lambda) \\
        &=\int_{\Omega} (K_{\lambda} \otimes \id)\kb{\eta_m}{\eta_m} (K_{\lambda}^* \otimes \id) \, d\mu(\lambda).
    \end{align} 
    Define then the measures  $\nu_m$ by $\nu_m(X):=\int_X \tr{(K_{\lambda} \otimes \id)\kb{\eta_m}{\eta_m} (K_{\lambda}^* \otimes \id)} \, d\mu(\lambda)$. These are probability measures, as in the domain $\hD$, $\int_{\Omega} K_{\lambda}^* K_{\lambda} \, d\mu(\lambda)$  is equal to the identity, and each $\ket{\eta_m}$ consist of vectors belonging to this domain: $\nu_m(\Omega)=\ip{\eta_m}{\int_{\Omega} K_{\lambda}^* K_{\lambda} \, d\mu(\lambda) \otimes \id| \eta_m}=1$. Also $\tr{(K_{\lambda} \otimes \id)\kb{\eta_m}{\eta_m} (K_{\lambda}^* \otimes \id)}>0$ $\nu_m$-almost everywhere. Then 
    \begin{align}
        (\Lambda \otimes \id)(\kb{\eta_m}{\eta_m})&=\int_{\Omega} (K_{\lambda} \otimes \id)\kb{\eta_m}{\eta_m} (K_{\lambda}^* \otimes \id) \, d\mu(\lambda) \\
        &=\int_{\Omega} \frac{(K_{\lambda} \otimes \id)\kb{\eta_m}{\eta_m} (K_{\lambda}^* \otimes \id)}{\tr{(K_{\lambda} \otimes \id)\kb{\eta_m}{\eta_m} (K_{\lambda}^* \otimes \id)}} \, d\nu_m(\lambda)
    \end{align}
    One can then easily see that since $\mathrm{rank}(K_{\lambda})\leq n$, for the pure states $\kb{\psi_{m\lambda}}{\psi_{m\lambda}}:=\frac{(K_{\lambda} \otimes \id)\kb{\eta_m}{\eta_m} (K_{\lambda}^* \otimes \id)}{\tr{(K_{\lambda} \otimes \id)\kb{\eta_m}{\eta_m} (K_{\lambda}^* \otimes \id)}}$ we also have $\mathrm{SR}(\kb{\psi_{m\lambda}}{\psi_{m\lambda}})\leq n$. Therefore by Proposition 1 B) in \cite{shirokov2013schmidt} we have that $(\Lambda \otimes \id)(\kb{\eta_m}{\eta_m}) \in \hS_n(\ki \otimes \hi)$ for all $m$. \par 
    Now as the channel $\Lambda \otimes \id$ is trace-norm continuous, we have $\Lambda \otimes \id(\kb{\eta}{\eta})=\lim_{m\to \infty} \Lambda \otimes \id (\kb{\eta_m}{\eta_m})$ (in the case $\dim\hi=\infty$).  Since $\hS_n(\ki \otimes \hi)$ is closed in the trace norm, we have that $\Lambda \otimes \id(\kb{\eta}{\eta}) \in \hS_n(\ki \otimes \hi)$. Thus by Proposition 8 in \cite{shirokov2013schmidt} $\Lambda$ is $n$-PEB.
    \end{proof}
\begin{proposition}\label{thmjointmeas}
   The set $\hM:=\{M_x\}_x$ is jointly measurable if and only if it is 1-simulable in the sense of Definition \ref{defsimulability}.
\end{proposition}
\begin{proof}
    ($\Rightarrow$) Say the POVMs $M_x$ operate on the separable Hilbert space $\hi$. Now by assumption there exists a POVM $G$ defined on ($\Omega,\hA$), and Markov kernels $p(\cdot |x,\cdot)$ such that,  for all $x$ and $X_x$, 
    \begin{align}
        M_x(X_x)=\int_{\Omega} p(X_x|x, \lambda) \, dG(\lambda). \label{jm2}
    \end{align}
    Here we may replace $G$ with its rank-1 refinement $G^1$ without loss of generality \cite{Pello9}. Therefore by Theorem 1 in \cite{Pello5} there exists a measurable\footnote{That is, the maps $\lambda\mapsto\ip{d(\lambda)}\fii$ are $\mu$-measurable.} map $\lambda \mapsto \bra{d(\lambda)}$, where each $\bra{d(\lambda)}$ is an linear functional defined on a common dense subspace $\hD$ of $\hi$, such that the following holds for all $X \in \hA$ and $\fii,\psi \in \hD$:
    \begin{align}
        \ip{\fii}{G(X)|\psi}=\int_{X} \ip{\fii}{d(\lambda)}\ip{d(\lambda)}{\psi} \, d\mu(\lambda). \label{rank1}
    \end{align}
    Here $\mu$ is a probability measure, and the notations $\bra{d(\lambda)}(\fii)=\ip{d(\lambda)}\fii=\overline{\ip{\fii}{d(\lambda)}}$ are used. Then by combining equations (\ref{jm2}) and (\ref{rank1}) we get, for all $\fii,\psi \in \hD$,
    \begin{align}
        \ip{\fii}{M_x(X_x)|\psi}=\int_{\Omega} p(X_x|x, \lambda) \ip{\fii}{d(\lambda)}\ip{d(\lambda)}{\psi} \, d\mu(\lambda).
    \end{align}
    Let then $\eta \in \hi$ ($=:\ki$) be a unit vector and define (possibly unbounded) operators $K_{\lambda}:=\kb{\eta}{d(\lambda)}$ which are obviously rank-1. For $N_{x,\lambda}(\cdot):=p(\cdot|x,\lambda)\id$ we have the following:
    \begin{align}
        \ip{\fii}{M_x(X_x)|\psi}=\int_{\Omega} \ip{\fii}{K_{\lambda}^*N_{x,\lambda}(X_x)K_{\lambda}|\psi} \, d\mu(\lambda).
    \end{align}
    \par 
    ($\Leftarrow$) Since the operators $K_{\lambda}$ are now rank-1, the image of $K_{\lambda}$ is $\C \fii_{\lambda}$ for some unit vector $\fii_{\lambda}$ in a Hilbert space $\ki$. Define then the linear functional $\bra{d(\lambda)}$ by $\ip{d(\lambda)}{\fii}=\ip{\fii_{\lambda}}{K_{\lambda}|\fii}$ for all $\fii \in \hD$. Then the integral formula reduces to the following.
    \begin{align}
        \ip{\fii}{M_x(X_x)|\psi}=\int_{\Omega} \ip{\fii_{\lambda}}{N_{x,\lambda}(X_x)|\fii_\lambda} \ip{\fii}{d(\lambda)}\ip{d(\lambda)}{\psi} \, d\mu(\lambda) \label{jm3}
    \end{align}
    The densely defined integrals $G(X):=\int_{X} \kb{d(\lambda)}{d(\lambda)} \, d\mu(\lambda)$ (in the sandwich-sense) define a POVM and $p(X_x|x,\lambda):=\ip{\fii_{\lambda}}{N_{x,\lambda}(X_x)|\fii_\lambda}$ defines a Markov kernel. Combining these into equation (\ref{jm3}) we see that $\{M_x\}_x$ is jointly measurable. 
\end{proof}

\begin{proposition}\label{thmsimulabilitydiscrete}
    Let $\hM:=\{M_{a|x}\}_{a,x}$ be a set of discrete POVMs defined in a finite-dimensional Hilbert space with a fixed finite amount of outcomes $a$ and settings $x$. $\hM$ is $n$-simulable in the sense of Definition \ref{defsimulability} if and only if it is $n$-simulable according to Eq.\ (\ref{Eq:simfinite}).
\end{proposition}
\begin{proof}
    We postpone the proof by noting that it can be seen as a consequence of Proposition \ref{thmnprepconsistency}.
\end{proof}

Finally we show as an example that the spectral measures of the canonical pair of position and momentum are not $n$-simulable for any finite $n$, i.e.\ position and momentum are not \emph{finitely simulable}. The spectral measure of the momentum operator, $P$ is related to position via the Fourier-Plancherel operator $F$, i.e.\ the extension of the Fourier transform to $L^2(\R)$, with the formula $P(X)=F^*Q(X)F$ for all $X \in \hB(\R)$. 
\begin{example} \label{thmQPsimulable}
    The set $\{Q,P\}$ is not finitely simulable.
    Let us assume that in some dense domain $\hD \subseteq \hi$ the following two equations hold in the weak sense for all $Y \in \hB(\R)$:
\begin{align}
Q(Y)&=\int_{\Omega} K_{\lambda}^* N_{1,\lambda}(Y) K_{\lambda} \, d\mu(\lambda), \\
P(Y)&=\int_{\Omega} K_{\lambda}^* N_{2,\lambda}(Y) K_{\lambda} \, d\mu(\lambda).
\end{align}
We (densely) define the following POVM on the domain $\hA$ of $\mu$:
\begin{align}
E(Z):=\int_{Z} K_{\lambda}^* K_{\lambda} \, d\mu(\lambda).
\end{align}
Then we have a densely defined joint observable $G_1$ for $Q$ and $E$ defined by
\begin{align}
G_1(Y \times Z):=\int_{Z}  K_{\lambda}^* N_{1,\lambda}(Y) K_{\lambda} \, d\mu(\lambda).
\end{align}
Therefore $Q$ and $E$ are jointly measurable and as $Q$ is a spectral measure, they must also commute. Similarily, with the joint observable 
\begin{align}
G_2(Y \times Z):=\int_{Z}  K_{\lambda}^* N_{2,\lambda}(Y) K_{\lambda} \, d\mu(\lambda)
\end{align}
we deduce that $E$ and $P$ must commute. Therefore $E$ commutes with the generators of the Weyl group, which is defined by $W(q,p):=e^{iqp/2}U(q)V(p)$, where $(U(q)\fii)(x)=\fii(x-q)$ and $(V(p)\fii)(x)=e^{ipx}\fii(x)$, $\fii\in L^2(\R)$. Therefore $E$ commutes with the entire Weyl group, which is irreducible, so $E(Z)=p(Z)\id$, $Z\in\hA$, for some probability measure $p$. \par 
 Let then $0\neq \fii \in \hD$ be an arbitrary vector. Then
\begin{align}
p(Z)\Vert \fii \Vert^2=\int_{Z} \Vert K_{\lambda}\fii \Vert^2 \, d\mu(\lambda).
\end{align}
Thus for any two vectors $\fii,\psi \in \hD$ and for all sets $Z \in \hA$ we get that 
\begin{align}
\int_{Z} \frac{\Vert K_{\lambda}\fii \Vert^2}{\Vert \fii \Vert^2} \, d\mu(\lambda)=\int_{Z} \frac{\Vert K_{\lambda}\psi \Vert^2}{\Vert \psi \Vert^2} \, d\mu(\lambda). \label{QP1}
\end{align}
Let then $\hB:=\{\fii_n\}_{n \in \N}$ be an orthonormal basis inside $\hD$. Define $V$ as the countable set of complex rational linear combinations of $\hB$. Let us use some numbering of $V=\{\psi_1,\psi_2,\dots \}$. Then if $m,n \in \N$, we have by equation (\ref{QP1}) that $\frac{\Vert K_{\lambda}\psi_m \Vert}{\Vert \psi_m \Vert}=\frac{\Vert K_{\lambda}\psi_n \Vert}{\Vert \psi_n\Vert}$ for all $\lambda \in \Omega\setminus N_{m,n}$ for some $N_{m,n}$ with $\mu(N_{m,n})=0$. Then also for the set $N:=\bigcup_{m,n=1}^{\infty} N_{m,n}$ we have $\mu(N)=0$. Therefore, for almost every $\lambda$ the equalities $\frac{\Vert K_{\lambda}\psi_m \Vert}{\Vert \psi_m \Vert}=\frac{\Vert K_{\lambda}\psi_n \Vert}{\Vert \psi_n\Vert} $ hold for all $m,n \in \N$. Thus for a fixed $\lambda \in \Omega \setminus N$ we have, for all $n \in \N$, 
\begin{align}
    \Vert K_{\lambda}\psi_n \Vert=A_{\lambda} \Vert \psi_n \Vert  \label{QP2}
\end{align}
where $A_{\lambda}\ge 0$ is a constant. Since $\fii_n,\fii_m \in \hB\subseteq V$ we get,  by the polarization identity and equation (\ref{QP2}), the following for almost every $\lambda \in \Omega$:
\begin{align}
    \ip{K_{\lambda}\fii_n}{K_{\lambda}\fii_m}=\frac{1}{4}\sum_{k=0}^3 i^k \Vert K_{\lambda} (\fii_m+i^k \fii_n) \Vert^2 =\frac{A_{\lambda}^2}{4}\sum_{k=0}^3 i^k \Vert  \fii_m+i^k \fii_n \Vert^2=A_{\lambda}^2\ip{\fii_n}{\fii_m}=A_{\lambda}^2 \delta_{nm}
\end{align}
Hence, for almost every $\lambda \in \Omega$ the image $ K_{\lambda}(\hD)$ contains an infinite orthogonal set (or $K_\lambda=0$). Thus for almost every $K_{\lambda}$ we have $\mathrm{rank}(K_{\lambda})=\infty$ (or 0).
\end{example}

\section{Application to quantum steering}

We now study the implications of the results concerning simulability on the concept of \emph{quantum steering}. In quantum steering, the objects of interest are called \emph{state assemblages}. To define these, we work in outcome spaces $(\Omega_x,\hA_x)$, where $x$ denotes a setting. Now a state assemblage is a collection $\{\sigma_x\}_x$ of positive trace-class valued measures $\sigma_x:\hA_x \to \th$ with the nonsignalling condition $\sigma_x(\Omega_x)=\sigma \in \sh$ for all $x$. With discrete $\sigma$-algebras these can be represented by collections $\{\sigma_{a|x}\}_{a,x}$, where $a$ denotes an outcome, with $\sum_{a} \sigma_{a|x}=\sigma$. Furthermore, for every state assemblage there exists a set of POVMs $\{A_x\}$ and a state $\rho \in \hS(\hi_A \otimes \hi_B)$ such that $\sigma_x(X_x)=\mathrm{tr}_A[(A_x(X_x) \otimes \id )\rho]$ for all $x$ and $X_x \in \hA_x$ i.e.\ every state assemblage can be prepared with a state $\rho$ and measurements $A_x$ \cite{gisin89,hughston93}. More concretely, we have two parties, Alice and Bob (with the corresponding Hilbert spaces $\hi_A$ and $\hi_B$ respectively), sharing a state $\rho$. Alice performs a local measurement $A_x$, based on Bob's classical communication of the setting $x$, on the state and communicates the result to Bob. Consequently, Bob's subsystem can be described by an unnormalized state $\sigma_x(X_x)=\mathrm{tr}_A[(A_x(X_x) \otimes \id )\rho]$. Here $\mathrm{tr}_A$ denotes the partial trace over the Hilbert space $\hi_A$. \par 
A state assemblage $\{\sigma_x\}_x$ acting on the Hilbert space $\hi$, defined in the measurable spaces $(\Omega_x,\hA_x)$ is deemed \emph{unsteerable}, when it admits a local hidden state (LHS) model. Formally this means that there exists a positive trace-class valued measure $T$, defined on the measurable space $(\Omega',\hA')$ such that $T(\Omega') \in \sh$, and post-processings $p(\cdot |x, \cdot):\hA_x \times \Omega' \to [0,1]$ such that for all $x$ and $X_x \in \hA_x$ we have 
\begin{align}
    \sigma_x(X_x)=\int_{\Omega'} p(X_x|x,\lambda) \, dT(\lambda).
\end{align}

A state $\rho \in \hS(\hi \otimes \ki)$ is called separable, if $\rho \in \overline{\mathrm{conv}(\{\kb{\fii}{\fii} \otimes \kb{\psi}{\psi} \, \, | \, \, \fii \in \hi, \, \psi \in \ki\}}$ i.e.\ it belongs to the trace-norm closure of the convex hull of pure product states. 
In the discrete case it has been proven that the existence of a LHS model for a given assemblage is equivalent to preparability by of this assemblage by some separable state \cite{kogias15b, moroder16}. Remarkably, however, this notion of preparability for unsteerable assemblages is more complex for the case of continuous variable steering in infinite dimensions: there are unsteerable state assemblages that need to be prepared with a non-separable state. This essentially follows from the example after Theorem \ref{thmseparability}, since every unsteerable assemblage $\{\sigma_x\}$ with the total state $\sigma$ corresponds to a jointly measurable set of POVMs $\{M_x\}$ such that $\sigma_x(X_x)=\sigma^{1/2}M_x(X_x)\sigma^{1/2}$. This correspondence is proven in the discrete case in \cite{uola15} and discussed in the continuous variable case in \cite{kiukas17}. The result is formalized below.
\begin{theorem}\label{thmnonseparable}
    Let $Q_1=Q_2=Q$ be the spectral measure of the position operator. Define the unsteerable state assemblage $\{\sigma_1,\sigma_2\}$, with $\sigma_x(X):=\sigma^{1/2}Q_x(X) \sigma^{1/2} $ for all $X \in \bo\R$ and $x=1,2$, where $\sigma$ is a (faithful) full-rank state. Then there exist no combination of a Hilbert space $\hi_A$ (even non-separable), separable state $\rho$ and set of POVMs $\{M_x\}_x$ operating in $\hi_A$ such that the following equation holds for $x \in \{1,2\}$:
    \begin{align}
        \sigma_x( \cdot )=\mathrm{tr}_A[(M_x( \cdot) \otimes \id)\rho].
    \end{align}
\end{theorem}
\begin{proof}
    We begin the proof by first proving a lemma. 
    \begin{lemma}
        Let $\hi_B$ be a separable Hilbert space and $\{\sigma_x\}_x$ a state assemblage such that $\sigma_x(X_x) \in \hT(\hi_B)$ for all $x$ and $X_x$. Assume that $\{\sigma_x\}_x$ can be prepared with a separable state $\rho \in \hS(\hi_A \otimes \hi_B)$ and a POVM acting in $\hi_A$, where $\hi_A$ is an arbitrary (possibly non-separable) Hilbert space. Then $\{\sigma_x\}_x$ can also be prepared with a separable state $\rho_0 \in \hS(\hi_{A0} \otimes \hi_B)$ and a POVM acting in $\hi_{A0}$, where $\hi_{A0}$ is a separable Hilbert space.
    \end{lemma}
    \emph{Proof of Lemma 1.} 
    We now have that for some POVMs $N_x$ and separable state $\rho$ the following holds for all $x$ and $X_x$:
    \begin{align}
        \sigma_x(X_x)=\mathrm{tr}_A[(N_x(X_x) \otimes \id)\rho].
    \end{align}
    Now by the spectral theorem \cite[Theorem 3.5]{busch16}, we have 
    \begin{align}
        \rho=\sum_{k} \lambda_k \kb{\fii_k}{\fii_k}
    \end{align}
    where $\{\fii_k\}$ is a countable orthonormal set, $\lambda_k\ge 0$, and $\sum_k\lambda_k=1$.
    Let us fix an orthonormal basis $\{\psi_i\}_{i \in I}$ for $\hi_A$, where $I$ is some index set, and a basis $\{\eta_n\}_{n=1}^{\dim\hi_B}$ for $\hi_B$. Then $\{\psi_i \otimes \eta_n\}_{i,n}$  is an orthonormal basis for $\hi_A \otimes \hi_B$. Therefore 
    \begin{align}
        \fii_k=\sum_{i \in C_k, n} a_{kin} \ket{\psi_i \otimes \eta_n},\qquad a_{kin}\in\C,\quad\sum_{i \in C_k, n} |a_{kin}|^2=1.
    \end{align}
    Here $C_k \subseteq I$ is a countable set for all $k$, since basis representations have to be summable. With this we get a new representation for $\rho$ (for example in the weak operator topology). 
    \begin{align}
        \rho=\sum_{k}\sum_{i,j \in C_k}\sum_{n,m} \lambda_k \overline{a_{kin}}a_{kjm} \kb{\psi_i \otimes \eta_n}{\psi_j \otimes \eta_m}.
    \end{align}
    Let us then define $\Tilde{C}_1:=C_1$ and $\Tilde{C}_k:=C_k\setminus \bigcup_{j=1}^{k-1}{C_j}$ for $k\geq 2$. Using this we (weakly) define the projection $P=\sum_{k}\sum_{i \in \Tilde{C}_k} \kb{\psi_i}{\psi_i}$. Then one can easily check that $(P \otimes \id)\rho(P \otimes \id)=\rho$. \\
    Next we define $\hi_{A0}:= P(\hi_A)$. This is a separable Hilbert space, since it has the countable orthonormal basis $\{\psi_i\}_{i \in \bigcup_{k=1}^{\infty} \Tilde{C_k}}$. Also obviously for all $x$ we have that $N_x^0:=PN_xP$ are POVMs operating in $\hi_{A0}$. Let then $B \in \li(\hi_B)$ be arbitrary. Now
    \begin{align}
        \mathrm{tr}[B\sigma_x(X_x)]&=\mathrm{tr}_{\hi_B}[B\mathrm{tr}_A[(N_x(X_x) \otimes \id) \rho]]=\mathrm{tr}_{\hi_A \otimes \hi_B}[(P \otimes \id)(P \otimes B)(PN_x(X_x)P \otimes \id) \rho (P \otimes \id)] \\
        &=\mathrm{tr}_{\hi_{A0} \otimes \hi_B}[(\id_{\hi_{A0}} \otimes B)(N_x^0(X_x) \otimes \id) \rho ] =\mathrm{tr}[B \mathrm{tr}_{A0}[(N_x^0(X_x) \otimes \id) \rho]].
    \end{align}
    Therefore $\{\sigma_x\}_x$ can indeed be prepared using a separable Hilbert space. All that is left to check is that the projection by $P$ preserves separability.  \par 
    By separability of $\rho$, there exists a sequence $(\rho_n)_{n \in \N}$ of convex combinations of product states such that $\lim_{n \to \infty} \Vert \rho -\rho_n \Vert_1=0$. Here $\Vert \cdot \Vert_1$ denotes the trace norm. Since  $(P \otimes \id)\rho (P \otimes \id)=\rho$, we easily see that
    \begin{align}
        \frac{(P \otimes \id) \rho_n (P \otimes \id)}{\tr{(P \otimes \id) \rho_n (P \otimes \id)}} \to \rho
    \end{align}
    in the trace norm of $\hT(\hi_A \otimes \hi_B)$. We may assume without loss of generality that for no $n \in \N$ we have $(P \otimes \id) \rho_n (P \otimes \id)=0$. This is since if some of these are zero, then we may consider the subsequence that contain no zeroes, since there can only be finitely many $(P \otimes \id) \rho_n (P \otimes \id)$ that are zero. This in turn follows from the fact that the sequence $((P \otimes \id) \rho_n (P \otimes \id))_n$ converges to a nonzero operator. \par   The states $\rho_{n0}:=\frac{(P \otimes \id) \rho_n (P \otimes \id)}{\tr{(P \otimes \id) \rho_n (P \otimes \id)}}$ are still obviously convex combinations of product states and also converge to $\rho$ in the trace norm of $\hT(\hi_{A0} \otimes \hi_B)$. Thus $\rho$ is separable in $\hS(\hi_{A0} \otimes \hi_B)$, which finishes the proof.\qed \\

    Let us then move on to the proof of Theorem \ref{thmnonseparable}. Assume that $\{\sigma_x\}_{x=1,2}$ is as in the statement of the theorem. Aiming for a contradiction, assume that $\{\sigma_x\}_x$ can be prepared with a separable state $\Tilde{\rho} \in \hS(\hi_A \otimes L^2(\R))$. By Lemma 1 we can assume that $\hi_A$ is separable and
    \begin{align}
        \sigma^{1/2}Q_x(X) \sigma^{1/2}=\mathrm{tr}_A[(M_x( X) \otimes \id)\Tilde{\rho}]. \label{nonsep1}
    \end{align}
    If (\ref{nonsep1}) holds, then there is also another separable state  $\rho $ such that $\sigma^{1/2}Q_x(X) \sigma^{1/2}=\mathrm{tr}_A[(M_x( X) \otimes \id)\rho]^{\rm T}$, where the transpose is in the eigenbasis of $\sigma$.
    Let $\sigma=\sum_{k=1}^{\infty} p_k \kb{k}{k}$ be the spectral representation of $\sigma$ and $\psi:=\sum_{k=0}^{\infty} \sqrt{p_k}\ket{k}\otimes\ket{k}$ its purification. 
    As $\rho$ is a state in a separable bipartite space, we can use the Choi-Jamiolkowski isomorphism for EB-channels \cite[Corollary 1]{shirokov2013schmidt} to deduce that $\rho=\Lambda \otimes \id( \kb{\psi}{\psi})$ for some EB-channel $\Lambda$. Let $B \in \li(L^2(\R))$ be arbitrary. Then we have
    \begin{align}
        \tr{B\sigma^{1/2}Q_x(X) \sigma^{1/2}}&=\tr{B\mathrm{tr}_A[(M_x( X) \otimes \id)\Lambda \otimes \id(\kb{\psi}{\psi})]^T}=\tr{B^T\mathrm{tr}_A[(M_x( X) \otimes \id)\Lambda \otimes \id(\kb{\psi}{\psi})]} \\
        &=\tr{(\Lambda^*(M_x( X)) \otimes B^T)\kb{\psi}{\psi}]} 
        = \tr{B \sigma^{1/2}\Lambda^*(M_x( X)) \sigma^{1/2}}.
    \end{align}
    Since the range of $\sigma^{1/2}$ is dense, we have that $Q_x(X)=\Lambda^*(M_x( X))$. But according to  the example after Theorem \ref{thmseparability}, this implies that $\hi_A$ cannot be separable, which is a contradiction. Hence, $\{\sigma_x\}_x$ cannot be prepared with a separable state.
\end{proof}

This shows that the claim made about preparability of unsteerable state assemblages with separable states in the continuous-variable setting, given in \cite{kogias15b}, needs to be slightly adjusted. Indeed, we prove in the following that to describe an equivalent condition to unsteerability, we need to add ``continuous mixing'' to the notion of preparability with separable states i.e.\ a barycenter. This result is to be expected in light of Proposition \ref{thmjointmeas} and the parallelism of joint measurability and unsteerability. Consequently the proof is also very similar. This is formalized below.

\begin{proposition}\label{thmunsteerable}
    A state assemblage $\{\sigma_x\}_x$ is unsteerable if and only if there is a probability space $(\Omega',\hA',\mu)$, and a weakly measurable collection of state assemblages $\{\sigma_x^{\lambda}\}_{x,\lambda}$, preparable with separable states, such that, for all $x$ and $X_x \in \hA_x$, 
    \begin{align}
        \sigma_x(X_x)=\int_{\Omega'} \sigma_x^{\lambda}(X_x) \, d\mu(\lambda). \label{US1}
    \end{align}
\end{proposition}
\begin{proof}
    ($\Rightarrow$) By unsteerability we have 
    \begin{align}
        \sigma_x(X_x)=\int_{\Omega'} p(X_x|x,\lambda) \, dT(\lambda).
    \end{align}
    As in the proof of Proposition \ref{thmjointmeas}, we can assume that $T$ is rank-1, by replacing it with its rank-1 refinement if necessary. Then 
    \begin{align}
        T(X)=\int_{X} \kb{\fii_{\lambda}}{\fii_{\lambda}} \, d\mu(\lambda)
    \end{align}
    for some measurable family $\{\fii_{\lambda}\}_{\lambda \in \Omega'} \subseteq \hi$ and a probability measure $\mu$. Thus we have 
    \begin{align}
        \sigma_x(X_x)=\int_{\Omega'} p(X_x|x,\lambda) \kb{\fii_{\lambda}}{\fii_{\lambda}} \, d\mu(\lambda)=\int_{\Omega'} p(X_x|x,\lambda) \frac{1}{\tr{\kb{\fii_{\lambda}}{\fii_{\lambda}}}}\kb{\fii_{\lambda}}{\fii_{\lambda}} \, d\nu(\lambda).
    \end{align}
    Here we defined a new probability measure by $\nu(X):=\int_X\tr{\kb{\fii_{\lambda}}{\fii_{\lambda}}} d\mu(\lambda)$. 
    Now the state assemblages $\sigma_{x}^{\lambda}$ defined with $\sigma_{x}^{\lambda}(X_x):=p(X_x|x,\lambda) \frac{1}{\tr{\kb{\fii_{\lambda}}{\fii_{\lambda}}}}\kb{\fii_{\lambda}}{\fii_{\lambda}}$ are preparable with a separable state since 
    $$\sigma_{x}^{\lambda}(X_x)=\mathrm{tr}_A\left[(p(X_x|x,\lambda)\id \otimes \id)\kb{\psi}{\psi} \otimes \frac{\kb{\fii_{\lambda}}{\fii_{\lambda}}}{\tr{\kb{\fii_{\lambda}}{\fii_{\lambda}}}}\right]$$
    for any unit vector $\psi$ in some Hilbert space $\hi_A$. \\
    ($\Leftarrow$) Assume equation (\ref{US1}) holds. We can assume that the $\sigma_x^{\lambda}$ can be prepared with a pure separable state. Then $\sigma_x^{\lambda}(X_x)=\mathrm{tr}_A[(N_{x,\lambda}(X_x) \otimes \id)\kb{\psi_{\lambda}}{\psi_{\lambda}} \otimes \kb{\fii_{\lambda}}{\fii_{\lambda}}]=\ip{\psi_{\lambda}}{N_{x,\lambda}(X_x)|\psi_{\lambda}}\kb{\fii_{\lambda}}{\fii_{\lambda}}$. Therefore we define $p(X_x|x,\lambda):=\ip{\psi_{\lambda}}{N_{x,\lambda}(X_x)|\psi_{\lambda}}$ for all $x,\,\lambda$ and $X_x$ as well as $T(X):=\int_{X}  \kb{\fii_{\lambda}}{\fii_{\lambda}} \, d\mu(\lambda)$ for all $X$. Plugging these into equation (\ref{US1}) we see that the state assemblage is unsteerable.  
\end{proof}
Now as in the case of simulability, we can use this result to generalize the concept of \emph{$n$-preparability}. This concept was first introduced in \cite{designolle21} in the discrete finite dimensional case. Indeed, a state assemblage $\{\sigma_{a|x}\}_{a,x} \subseteq \li(\C^d)$ is $n$-preparable if and only if there exists a state $\rho \in \hS(\hi_A \otimes \C^d)$ with $\mathrm{SN}(\rho)\leq n\leq d$ and a set of POVMs $\{A_{a|x}\}_{a,x}$ such that $\sigma_{a|x}=\mathrm{tr}_A[(A_{a|x} \otimes \id)\rho]$ for all $a$ and $x$. We thus see that in the finite dimensional discrete case, 1-preparability is exactly unsteerability and consequently we require this same property from any generalization of $n$-preparability, similar to that what we required from any generalization of simulability. Since not all unsteerable states are preparable with a separable state, the concept of $n$-preparability cannot be generalized directly in the case of continuous outcomes and infinite dimensions.  Instead, motivated by Proposition \ref{thmunsteerable}, we propose the following definition for $n$-preparability.

\begin{definition}\label{defpreparability}
    A state assemblage $\{\sigma_x\}_x$ is $n$-preparable if and only if there is a probability space $(\Omega',\hA',\mu)$, and a weakly measurable collection of state assemblages $\{\sigma_x^{\lambda}\}_{x,\lambda}$, preparable with states with Schmidt number at most $n$, such that, for all $x$ and $X_x \in \hA_x$, 
    \begin{align}
        \sigma_x(X_x)=\int_{\Omega'} \sigma_x^{\lambda}(X_x) \, d\mu(\lambda). \label{eqdefprep}
    \end{align}
\end{definition}

In the case of finite dimensional steering with discrete outcomes, the concepts of $n$-preparability and $n$-simulability were shown to be equivalent \cite{jones2022equivalence} in the sense that a state assemblage with a full rank total state is $n$-preparable if and only if the unique set of POVMs corresponding to this assemblage is $n$-simulable. This equivalence also holds for our general definitions for $n$-simulability and $n$-preparability, which is seen as follows. Assume first that the state assemblage $\{\sigma_x\}_x$ is $n$-preparable with the faithful full rank state $\sigma$ and corresponding set of POVMs $\{M_x\}_x$. Then we have 
\begin{align}
    \sigma^{1/2} M_x(X_x) \sigma^{1/2}=\int_{\Omega'} \mathrm{tr}_A[(N_{x,\lambda} (X_x)\otimes \id)\rho_{\lambda}] \, d\mu(\lambda).
\end{align}
Here $\mathrm{SN}(\rho_{\lambda}) \leq n$. Above we may assume the $\rho_{\lambda}$ are pure, as every state with a $\mathrm{SN}(\rho_{\lambda}) \leq n$ can be written as a convex combination of pure states with Schmidt rank at most $n$ and these convex combinations can be absorbed into the integration. Let then $\rho_{\lambda}=\sum_{i,j=1}^n \sqrt{p_{i\lambda}p_{j\lambda}} \kb{\fii_{i \lambda}}{\fii_{j \lambda}} \otimes \kb{\psi_{i \lambda}}{\psi_{j \lambda}}$ be the Schmidt decomposition for a fixed $\lambda$. Defining then the operators $F_{\lambda}=\sum_{i=1}^n \sqrt{p_{i \lambda}}\kb{\fii_{i \lambda}}{\psi_{i \lambda}}$, we find easily that 
\begin{align}
    \sigma^{1/2} M_x(X_x) \sigma^{1/2}=\int_{\Omega'} F_{\lambda}^* (N_{x,\lambda}(X_x))^{{\rm T}_{\lambda}} F_{\lambda} \, d\mu(\lambda).
\end{align}
Here $A^{{\rm T}_{\lambda}}$ means transpose of the operator $A$ in the extension of the partial Schmidt basis $\{\ket{\fii_{i \lambda}}\}_i$. Since $\sigma$ is a full rank faithful state, its range is dense and therefore we can define the densely defined unbounded operators $K_{\lambda}:=F_{\lambda}\sigma^{-1/2}$. We obviously then have $\mathrm{rank}(K_{\lambda})\leq n$ and therefore see that $\{M_x\}_x$ is $n$-simulable. \par 
Conversely, if $\{M_x\}$ is $n$-simulable, then weakly
\begin{align}
    \sigma_x(X_x)=\sigma^{1/2} M_x(X_x) \sigma^{1/2}=\int_{\Omega'} \sigma^{1/2} K_{\lambda}^* N_{x,\lambda}(X_x) K_{\lambda} \sigma^{1/2} \, d\mu(\lambda),
\end{align}
for some possibly unbounded operators $K_{\lambda}$ with $\mathrm{rank}(K_{\lambda})\leq n$ that are chosen in such a way that their common domain contains the eigenvectors of $\sigma$ \cite{Pello5}.
Taking the trace of $\sigma=\int_{\Omega'} \sigma^{1/2} K_{\lambda}^* K_{\lambda} \sigma^{1/2} \, d\mu(\lambda)$ we see that
 almost each $K_{\lambda}\sigma^{1/2}$ extends to a Hilbert-Schmidt operator $F_{\lambda}$. Then as $\mathrm{rank}(F_{\lambda})\leq n$, we can use a decomposition $F_{\lambda}=\sum_{i=1}^n \kb{\fii_{i \lambda}}{\psi_{i\lambda}}$ where $\{\ket{\fii_{i\lambda}}\}_i$ and $\{\ket{\psi_{i\lambda}}\}$ are orthogonal sets. We can then define the vectors $\ket{\eta_{\lambda}}:=\sum_{i=1}^n \ket{\fii_{i \lambda}}\otimes \ket{\psi_{i\lambda}}$, so that $\kb{\tilde{\eta}_{\lambda}}{\tilde{\eta}_{\lambda}}:=\frac{\kb{\eta_{\lambda}}{\eta_{\lambda}}}{\tr{\kb{\eta_{\lambda}}{\eta_{\lambda}}}}$ is a pure state with Schmidt rank at most $n$. Finally defining a new probability measure by $\nu(X):=\int_{X} \tr{\kb{\eta_{\lambda}}{\eta_{\lambda}}} \, d\mu(\lambda)$ we find by direct calculation that 
\begin{align}
    \sigma_x(X_x)=\sigma^{1/2} M_x(X_x) \sigma^{1/2}=\int_{\Omega'} \mathrm{tr}_A[(N_{x,\lambda}(X_x))^{{\rm T}_{\lambda}} \otimes \id) \kb{\tilde{\eta}_{\lambda}}{\tilde{\eta}_{\lambda}}] \, d\nu(\lambda).
\end{align}
Therefore $\{\sigma_x\}_x$ is $n$-preparable. \par 
This equivalence further demonstrates the consistency of the general definitions of simulability and $n$-preparability. Furthermore, Proposition \ref{thmsimulabilitydiscrete} follows now from this equivalence and the following consistency result.
\begin{proposition}\label{thmnprepconsistency}
    Let $\{\sigma_{a|x}\}_{a,x}$ be a discrete state assemblage defined in a finite dimensional Hilbert space and with a fixed finite amount of outcomes $a$ and settings $x$ . If $\{\sigma_{a|x}\}_{a,x}$ is $n$-preparable in the sense of Definition \ref{defpreparability}, then it is also $n$-preparable according to the definition in \cite{designolle21}. In this case the decomposition (\ref{eqdefprep}) can be given with a discrete probability measure: $\sigma_{a|x}=\sum_{\lambda}p_{\lambda}\sigma_{a|x}^{\lambda}$ for all $a$ and $x$. \\
    Conversely, if $\{\sigma_{a|x}\}_{a,x}$ is $n$-simulable in the original definition, then it is also $n$-simulable according to Definition \ref{defpreparability}.
\end{proposition}
\begin{proof}
    Let us first assume that the discrete assemblage $\{\sigma_{a|x}\}_{a,x} \subseteq \li(\C^d)$ is $n$-preparable according to Definition \ref{defpreparability}. To prove that this can be realized with a discrete measure, we first prove the following lemma. 
    \begin{lemma}
        The set of $n$-preparable state assemblages with a fixed finite amount of outcomes $a$ and settings $x$ according to the definition given in \cite{designolle21} i.e.\ the set 
        \begin{align}
            \mathcal{SA}_n:=\left\{ \bigoplus_{a,x} \sigma_{a|x}  \in \bigoplus_{a,x} \li(\C^d) \, \, \bigg| \, \, \exists[\rho \in \hS(\hi_A \otimes \C^d), \, \mathrm{SN}(\rho)\leq n, \, \text{POVMs }\{N_{a|x}\}_{a}  ]: \sigma_{a|x}=\mathrm{tr}_A[(N_{a|x} \otimes \id)\rho]  \right\}
        \end{align}
        is compact in the equivalent norm topologies of $\bigoplus_{a,x} \li(\C^d) $.
    \end{lemma}
    \noindent\emph{Proof of Lemma 2.}
 Since we are dealing with finite dimensional normed spaces, it is enough to show that this set is closed and bounded. To show boundedness, let us use the trace-norm $\Vert \cdot \Vert_1$. Let then $\bigoplus_{a,x} \sigma_{a|x} \in  \mathcal{SA}_n$. Now boundedness follows from the following inequality.
    \begin{align}
        \Vert \sigma_{a|x} \Vert_1 \leq \left\Vert \sum_a \sigma_{a|x} \right\Vert_1=1
    \end{align}
    Let's then show that $\mathcal{SA}_n$ is closed. Let $\left( \bigoplus_{a,x} \sigma_{a|x,k}\right)_k$ be a convergent sequence in $\mathcal{SA}_n$ with the limit $ \bigoplus_{a,x} \sigma_{a|x}$ and $(\sigma_{a|x,k})_k$ the sequence of an arbitrary block. Obviously $ \bigoplus_{a,x} \sigma_{a|x}$ is a state assemblage, as limit preserves positivity and the non-signalling condition holds since the set of states is closed. We can assume that each $\sigma_{a|x,n}$ is prepared with a pure state $\kb{\psi_k}{\psi_k}$ with $\mathrm{SR}(\kb{\psi_k}{\psi_k})\leq n$. We now decompose these pure states in their Schmidt decomposition, 
    \begin{align}
        \kb{\psi_k}{\psi_k}=\sum_{i,j=1}^n \sqrt{p_{ki} p_{kj}} \kb{\fii_{ki} \otimes \psi_{ki}}{\fii_{kj} \otimes \psi_{kj}}.
    \end{align}
    Let us then extend the orthonormal vectors $\{\ket{\fii_{ki}}\}_{i=1}^n$ in the Schmidt decompositions to orthonormal bases $\{\ket{\fii_{ki}}\}_{i=1}^{\dim \hi_A}$ for all $k \in \N$. Define then the unitary operators $U_k \in \li(\hi_A)$, $k \in \N$ by the relations $U_k\ket{\fii_{1i}}=\ket{\fii_{ki}}$ for all $i$. Now for an arbitrary $k \in \N$ we have 
    \begin{align}
        \kb{\psi_k'}{\psi_k'}:=(U_k^* \otimes \id)\kb{\psi_k}{\psi_k}(U_k\otimes \id)=\sum_{ij=1}^n \sqrt{p_{ki} p_{kj}} \kb{\fii_{1i} \otimes \psi_{ki}}{\fii_{1j} \otimes \psi_{kj}}. \label{NPD1}
    \end{align}
    Let then $\hi_{A0} := \lin (\{\fii_{1i}\}_{i=1}^n) $. By equation (\ref{NPD1}) we can interpret that all $\kb{\psi_k'}{\psi_k'}$ are states in the finite dimensional Hilbert space  $\hi_{A0} \otimes \C^d$. Furthermore we have the following for some POVMs $N_{a|x,k}$:
    \begin{align}
        \sigma_{a|x,k}&=\mathrm{tr}_A[(N_{a|x,k} \otimes \id)\kb{\psi_k}{\psi_k}]=\mathrm{tr}_A[(U_k^*N_{a|x,k}U_k \otimes \id)(U_k^* \otimes I)\kb{\psi_k}{\psi_k}(U_k \otimes I)]\\
        &=[(U_k^*N_{a|x,k}U_k \otimes \id)(\kb{\psi_k'}{\psi_k'}].
    \end{align}
    Therefore for all $k \in \N$ we can assume the state assemblages given by $\sigma_{a|x,k}$ can be prepared with pure states defined in a common finite dimensional Hilbert space. Let us still denote the sequence of these states by $(\kb{\psi_k'}{\psi_k'})_k$. Since the set of pure states in a finite dimensional Hilbert space with Schmidt rank at most $n$ is compact in the trace norm, the sequence $(\kb{\psi_k'}{\psi_k'})_k$ has a convergent subsequence $(\kb{\psi_{k_j}'}{\psi_{k_j}'})_j$ with a limit $\kb{\psi}{\psi}$ for which still $\mathrm{SR}(\kb{\psi}{\psi})\leq n$. Also, since the partial trace is continuous, we have $\sum_{a} \sigma_{a|x} =\lim_{k \to \infty} \sum_{a}\sigma_{a|x,k}=\lim_{j\to \infty}\sum_{a}\sigma_{a|x,k_j} = \lim_{j\to \infty} \mathrm{tr_A}[\kb{\psi_{k_j}'}{\psi_{k_j}'}]=\mathrm{tr_A}[\kb{\psi'}{\psi'}]$. Denoting $\omega:=\mathrm{tr_A}[\kb{\psi'}{\psi'}]$ we finally find 
    \begin{align}
        \sigma_{a|x}=\mathrm{tr}_A[(\omega^{-1/2} \sigma_{a|x} \omega^{-1/2} \otimes \id)\kb{\psi'}{\psi'} ].
    \end{align}
    Thus $\bigoplus_{a,x} \sigma_{a|x} \in \mathcal{SA}_n$, and therefore $\mathcal{SA}_n$ is closed. Combining this with boundedness we see that $\mathcal{SA}_n$ is compact. \qed \\

    Let us then return to the proof of Proposition \ref{thmnprepconsistency}. For the $n$-preparable state assemblage $\{\sigma_{a|x}\}_{a,x}$ we now have $\sigma_{a|x}=\int_{\Omega'} \sigma_{a|x}^{\lambda} \, d\mu(\lambda)$, where $\bigoplus_{a,x} \sigma_{a|x}^{\lambda} \in \mathcal{SA}_n$. Now $\mathcal{SA}_n$ is convex \cite{uola19b} and by Lemma 2 it is compact in the trace-norm topology.  Therefore each $\bigoplus_{a,x}\sigma_{a|x}$ is a barycenter of a probability measure over a compact convex set and consequently $\bigoplus_{a,x}\sigma_{a|x} \in \mathcal{SA}_n$ \cite{Alfsen1971CompactCS}.  \\
    The converse direction is obvious. 
\end{proof}

As mentioned before, Proposition \ref{thmsimulabilitydiscrete} now follows easily from this result. As a final example we show that Example \ref{thmQPsimulable} implies also the existence of a state assemblage that can only be prepared with a state with infinite Schmidt number.

\begin{example}
Let $\sigma$ be a faithful full rank state, and define the state assemblage $\{\sigma_x\}_{x=1,2}$ with the following relations.
\begin{align}
    \sigma_1(X)&:=\sigma^{1/2}Q(X)\sigma^{1/2}, \\
    \sigma_2(X)&:=\sigma^{1/2}P(X)\sigma^{1/2}.
\end{align}
Aiming for a contradiction, assume that this state assemblage is $n$-preparable for some $n \in \N$. Using the equivalence between $n$-preparability and $n$-simulability, this means that $\{Q,P\}$ is $n$-simulable, which is a contradiction with Example \ref{thmQPsimulable}. Therefore this state assemblage cannot be $n$-preparable for any $n \in \N$.
\end{example}

\section{Conclusions}

We have presented a generalisation of a known simulation algorithm of quantum measurements to the continuous variable regime. The generalisation is based on instruments whose output space is not the full quantum algebra, hence, providing an operational task in which such general mappings are needed. We have demonstrated the need for the generalisation, given that one requires the compression protocol to be a relaxation of the central concept of joint measurability. We have further shown that the canonical pair of position and momentum is not finitely simulable.

By bridging our notion with bipartite quantum correlations, we have introduced a stronger form of quantum steering for the continuous variable setting. The resulting correlation tasks ask one to verify the Schmidt number of a continuous variable state using the steering setup. Translating the results from the measurement side, we have shown that not all unsteerable state assemblages can be prepared with a separable state, and that an analogue of the original EPR setting, in which one measures position and momentum on a full Schmidt rank shared state, is not $n$-preparable for any finite $n$.

For future research, it could be of interest to investigate what sort of a role covariance systems play in finding simulation protocols. Also, it would be interesting to develop a smoothened version of the generalised simulation protocol, in which one is not required to simulate the target measurements exactly, but would be allowed to simulate them approximately. Also, it could be of interest to investigate whether the known connections between incompatibility and an advantage in state discrimination tasks could be generalised in a meaningful way to the realm of simulability, in that the dimensionality would play a central role in the task \cite{carmeli19a,skrzypczyk19,oszmaniec19,uola19b,uola19c}.

\section{Acknowledgements}

PJ, SE, and RU are thankful for the support from  the Swiss National Science Foundation (Ambizione PZ00P2-202179). All authors are grateful for fruitful discussions with Erkka Haapasalo, Benjamin D.M. Jones, Jonathan Steinberg.

\bibliographystyle{unsrt}
\bibliography{references.bib}

\end{document}